%% file: main.tex
\documentclass[11pt]{article}

\usepackage{amssymb,amsmath,amsthm,amsfonts}
\usepackage[numbers]{natbib}
\usepackage{microtype,color}
\usepackage{enumitem}

\newcommand{\Omit}[1]{}

\newtheorem{define}{Definition}

\newtheorem{theorem}{Theorem}
\newtheorem{lemma}{Lemma}

\newcommand{\mysum}[2]{\displaystyle\sum\limits_{#1}^{#2}}

\newcommand{\Outcomes}{\mathcal{O}}
\newcommand{\outcome}{o}
\newcommand{\States}{\Omega}
\newcommand{\state}{\omega}
\newcommand{\prm}{^{\prime}}
\newcommand{\vh}{\hat{v}}
\newcommand{\ph}{\hat{p}}
\newcommand{\vv}{\vec{v}}
\newcommand{\pv}{\vec{p}}

\newcommand{\subStates}[2]{\States_{#1,#2}}
\newcommand{\ioStates}{\subStates{i}{\outcome}}
\newcommand{\ioprmStates}{\subStates{i\prm}{\outcome}}

\newcommand{\Sio}{S_{o,i,\pv_{-i}(o)}}
\DeclareMathOperator*{\cart}{\times}

\begin{document}
\title{Designing Markets for Daily Deals}
\author{
Yang Cai\thanks{MIT, ycai@csail.mit.edu.},
Mohammad Mahdian\thanks{Google, mahdian@google.com.},
Aranyak Mehta\thanks{Google, aranyak@google.com.}, and
Bo Waggoner\thanks{Harvard, bwaggoner@fas.harvard.edu.}
}

\maketitle

\begin{abstract}
Daily deals platforms such as Amazon Local, Google Offers, GroupOn, and
LivingSocial have provided a new channel for merchants to directly market
to consumers. In order to maximize consumer acquisition and retention,
these platforms would like to offer deals that give good value to users.
Currently, selecting such deals is done manually; however, the large number
of submarkets and localities necessitates an automatic approach to selecting
good deals and determining merchant payments.

We approach this challenge as a market design problem. We postulate that
merchants already have a good idea of the attractiveness of their deal to
consumers as well as the amount they are willing to pay to offer their deal.
The goal is to design an auction that maximizes a combination of the revenue of the auctioneer
(platform), welfare of the bidders (merchants), and the positive externality on a third
party (the consumer), despite the asymmetry of information
about this consumer benefit.
We design auctions that truthfully elicit this information from the merchants
and maximize the social welfare objective, and we characterize the consumer
welfare functions for which this objective is truthfully implementable.
We generalize this characterization to a very broad mechanism-design setting and give examples
of other applications.
\end{abstract}


\setlength{\abovedisplayskip}{7pt}
\setlength{\belowdisplayskip}{7pt}

\input{intro}
\input{model}
\input{mechanism}

\input{impossibility}
\input{general}
\input{conclusion}

\bibliographystyle{plain}
\bibliography{coupons}

\end{document}

%% file: intro.tex
\section{Introduction}
\label{sec:intro}

Daily deals websites such as Amazon Local, Google Offers, GroupOn, and
LivingSocial have provided a new channel of direct marketing for
merchants. In contrast to standard models of advertising such as
television ads and web search results, the daily deals setting provides
two new challenges to platforms.

First, in models of advertising such as web search, the advertisement
is shown on the side of the main content; in contrast, daily deals websites
offer consumers web pages or emails that contain only advertisements
(\emph{i.e.}, coupons). Therefore, for the long-term success of a platform,
the decision of which coupons to show
to the user must depend heavily on the benefit these coupons provide to consumers.

Second, the merchant often has significantly more information
than the advertising platform about this consumer benefit. This benefit
depends on many things: how much discount the
coupon is offering, how the undiscounted price compares with the price
of similar goods at the competitors, the price elasticity of demand for
the good, the fine prints of the coupon, and so on.
These parameters are known to the merchants, who
routinely use such information to optimize their pricing and their
inventory, but not to the platform provider who cannot be expected to
be familiar with all markets and would need to invest significant
resources to learn these parameters. Furthermore, unlike standard
advertising models where an ad is displayed over a time period to a
number of users and its value to the user (often measured using
proxies like click-through rate or conversion rate) can be estimated
over time, the structure of the daily deals market does not permit
much experimentation: A number of deals must be selected at the
beginning of each day to be sent to the subscribers all at
once, and the performance of previous coupons, if any, by the same
advertiser is not a good predictor of the performance of the current
coupon, as changing any of the terms of the coupons can significantly
affect its value.

These challenges pose a novel {\em market design} problem: How can we select deals
with good benefit to the consumer in the presence of strongly asymmetric
information about this benefit? This is precisely our goal in this paper. We postulate that merchants
hold, as private information, two parameters: A \emph{valuation}
equalling the overall utility the merchant gains from being selected
(as in a standard auction); and a \emph{quality} that represents
the attractiveness of their deal to a user.
The task is to design an auction mechanism that incentivizes the merchants
to reveal their private information about both their valuation and
quality, then picks deals that maximize a combination of
platform, merchant, and consumer values.
We show that, if consumer welfare is a convex function of
quality, then we can design a truthful auction that
maximizes total social welfare; furthermore, we show that the
convexity condition is necessary.
We give negative results for another natural goal, achieving a
constant-fraction welfare objective subject to a quality threshold guarantee.
The main idea behind our positive results is to design a mechanism where
bidders' total payment is contingent (in a carefully chosen way) upon whether
the consumer purchases the coupon. Not surprisingly, the theory of proper scoring
rules comes in handy here.

We then extend these results to characterize incentive-compatible mechanisms
for social welfare maximization in a very general auction setting, where
the type of each bidder has both a valuation and a quality component.
Quality is modeled as a distribution over possible states of the world;
a \emph{consumer welfare function} maps these distributions to the welfare of
some non-bidding party.
We design truthful welfare-maximizing mechanisms for this setting and
characterize implementable consumer welfare functions with a convexity
condition that captures expected welfare and, intuitively, risk-averse preferences.
We give a number of example applications demonstrating that our
framework can be applied in a broad range of mechanism design
settings, from network design to principal agent problems.

The rest of this paper is organized as follows: In the next section,
we formally define the setting and the problem. In
Section~\ref{sec:mechanism}, we give a mechanism for maximizing
social welfare when consumer welfare is a convex function the quality. In
Section~\ref{sec:impossibility}, we show that no truthful
mechanism even approximates the objective of maximizing the
winner's value subject to a minimum quality; we also show
that the convexity assumption in
Section~\ref{sec:mechanism} is necessary. Finally, in
Section~\ref{sec:general}, we extend our
mechanisms and characterization to a much more general setting.

\paragraph{Related work.}
To the best of our knowledge, our work is the first to address
mechanism design in a market for daily deals. There has been unrelated
work on other aspects of daily deals (\emph{e.g.} impact on
reputation)~\cite{mitzenmacher1,mitzenmacher2,lu}. A related,
but different line of work deals with mechanism design for pay-per-click
(PPC) advertising. In that setting, as in ours, each ad has a value and
a quality (representing click-through rate for PPC ads and the
probability of purchasing the deal in our setting). The objective is
often to maximize the combined utility of the advertisers and the
auctioneer~\cite{hal,eos}, but variants where the utility of the user
is also taken into account have also been
studied~\cite{userexperience}.  The crucial difference is that in
PPC advertising, the auctioneer holds the quality parameter,
whereas in our setting, this parameter is only known to the merchant
and truthful extraction of the parameter is an important part of the
problem. Other work on auctions with a quality
component~\cite{che,espinola} assume that a quality level may be
assigned by the mechanism to the bidder (who always complies), in
contrast to our setting where quality is fixed and private
information.

We make use of proper scoring rules, an overview of which appears
in~\cite{gneiting2007}; to our knowledge, proper scoring rules have
been used in auctions only to incentivize agents to guess others'
valuations~\cite{azar2012}.
Our general setting is related to an extension of proper scoring rules,
decision rules and decision markets~\cite{hanson1999,othman2010}. There,
a mechanism designer elicits agents' predictions of an event conditional
on which choice she makes. She then selects an outcome,
observes the event, and pays the agents according to the
accuracy of their predictions. Unlike our setting, agents are assumed
not to have preferences over the designer's choice, except
in~\cite{boutilier2012eliciting}, which (unlike us) assumes
that the mechanism has partial knowledge of these preferences and does
not attempt to elicit preferences.
Our general model may be interpreted as a fully general extension to the
decision-rule setting in which we introduce the novel challenge of
\emph{truthfully eliciting} these preferences and incorporate them into
the objective. However, we focus on deterministic
mechanisms, while randomized mechanisms have been shown to have nice
properties in a decision-rule setting~\cite{chen2011}.

Another related line of work examines examines when a proper
scoring rule might incentivize an agent to take undesirable actions in
order to improve his prediction's accuracy. When the mechanism
designer has preferences over different states, scoring rules that
incentivize beneficial actions are termed \emph{principal-aligned}
scoring rules~\cite{vince}. A major difference is that the mechanism
designer in the principal-aligned setting, unlike in ours, does not
select between outcomes of any mechanism, but merely observes a state
of the world and makes payments. 

%% file: model.tex
\section{The Model}
\label{sec:model}
In this section, we formulate the problem in its simplest form: when
an auctioneer has to select just one of the interested merchants to
display her coupon to a single consumer.\footnote{Our mechanisms for this
  model can be immediately extended to the case of many consumers by scaling.}
In Section~\ref{sec:general}, our model and results will be generalized to a much broader setting.

There are $m$ bidders, each with a single coupon. We also refer
to the bidders as {\em merchants} and to coupons as {\em deals}.  An
auctioneer selects at most one of these coupons to display.
For each bidder $i$, there is a probability $p_i \in [0,1]$ that
if $i$'s coupon is displayed to a consumer, it will be purchased by
the consumer. We refer to $p_i$ as
the {\em quality} of coupon $i$. Furthermore, for each bidder $i$,
there is a value $v_i \in \mathbb{R}$ that represents the expected
value that $i$ gets if her coupon is chosen to be displayed to the
advertiser. Both $v_i$ and $p_i$ are private information
of the bidder $i$, and are unknown to the auctioneer.\footnote{In
  Section \ref{subsection:examples-general}, we will
  briefly discuss extensions in which both parties have quality information.}
We refer to
$(v_i,p_i)$ as bidder $i$'s {\em type}. We assume that the bidders are
expected utility maximizers and their utility is quasilinear in payment.

Note that $v_i$ is $i$'s total expected valuation for being selected;
in particular, it is \emph{not} a value-per-purchase (as in \emph{e.g.} search
advertisement). Rather, $v_i$ is the maximum amount $i$
would be willing to pay to be selected (before observing the
consumer's purchasing decision). Also, we allow $v_i$ and $p_i$
to be related in an arbitrary manner.
If, for instance, $i$ derives value
$a_i$ from displaying the coupon plus an additional $c_i$
if the consumer purchases the coupon, then $i$ would compute
$v_i = a_i + p_i c_i$ and submit her true type $(v_i,p_i)$. For our results, we do
not need to assume any particular model of how $v_i$ is computed or
of how it relates to $p_i$.
 
An auction mechanism functions as follows. It asks each bidder $i$ to
reveal her private type $(v_i, p_i)$. Let $(\vh_i, \ph_i)$ denote the
type reported by bidder $i$. Based on these reports, the mechanism
chooses one bidder $i^*$ as the {\em winner} of the auction, i.e., the
merchant whose deal is shown.
Then, a consumer arrives; with probability $p_{i^*}$, she decides to
purchase the deal. Let $\state\in\{0,1\}$ denote the consumer's decision
(where $1$ is a purchase). The mechanism observes the consumer's decision
and then charges the bidders according to a payment rule, which may
depend on $\omega$.

We require the mechanism to be {\em truthful}, which means that it is,
first, {\em incentive compatible}:
for every merchant $i$ and every set of types reported by the other merchants,
$i$'s expected utility is maximized if she reports her true type
$(v_i,p_i)$; and second, {\em interim individually rational}:
each merchant receives a non-negative utility in expectation (over the
randomization involved in the consumer's purchasing decision) if she reports her true type.

The goal of the auctioneer is to increase some
combination of the welfare of all the parties involved. If we ignore the
consumer, this can be modeled by the sum of the utilities of the
merchants and the auctioneer, which, by quasi-linearity of the
utilities, is precisely $v_{i^*}$. To capture the welfare of the user,
we suppose that a reasonable proxy is the quality $p_{i^*}$ of the selected
deal.  We study two natural ways to combine the merchant/auctioneer welfare
$v_{i^*}$ with the consumer welfare $p_{i^*}$. One is to maximize
$v_{i^*}$ subject to the deal quality $p_{i^*}$ meeting a minimum
threshold $\alpha$. Another is to model the consumer's welfare as a function
$g(p_{i^*})$ of quality and seek to maximize total welfare $v_{i^*} + g(p_{i^*})$. In the latter
case, when $g$ is a convex function, we construct in the next section a
truthful mechanism that maximizes this social welfare function (and we show
in Section~\ref{sec:impossibility} that, when $g$ is not convex,
there is no such mechanism). For the former case, in
Section~\ref{sec:impossibility}, we prove that it is not possible to
achieve the objective, even approximately.

%% file: mechanism.tex
\section{A Truthful Mechanism via Proper Scoring Rules}
\label{sec:mechanism}

In this section, we show that for every {\em convex} function $g$,
there is an incentive-compatible mechanism that maximizes the social
welfare function $v_{i^*} + g(p_{i^*})$. A convex consumer welfare $g$
function may be natural in many settings. Most importantly, it includes the natural special case of a
linear function; and it also intuitively models \emph{risk aversion}, because (by definition
of convexity) the average welfare of taking a guaranteed outcome, which is
$p g(1) + (1-p) g(0)$, is larger than the welfare $g(p)$ of
facing a lottery over those outcomes.\footnote{\label{ft:risk-example}
  To see this, suppose $100$ consumers arrive, and the welfare of each is
  the convex function $g(p) = p^2$. If $50$ consumers see a deal with $p = 0$ and
  $50$ see a deal with $p = 1$, the total welfare is $50(0) + 50(1) = 50$.
  If all $100$ see a deal with $p = 0.5$, the total welfare is
  $100\left(0.5^2\right) = 25$. Under this welfare function, the
  ``sure bet'' of $50$ purchases is
  preferable to the lottery of 100 coin flips.}\footnote{\label{ft:risk-concave}
  Note that risk aversion is often associated
  with \emph{concave} functions. These are unrelated as they do \emph{not}
  map probability distributions to welfare; they are functions
  $u: \mathbb{R} \to \mathbb{R}$ that map
  wealth to welfare. Concavity represents risk aversion in that setting because
  the welfare of a guaranteed payoff $x$, which is $u(x)$, is larger than the
  welfare of facing a draw from a distribution with probability $x$, which is
  $x u(1) + (1-x)u(0)$.}

We will make use of \emph{binary scoring rules}, which are defined as follows.

\begin{define}
A binary scoring rule $S:[0,1]\times\{0,1\}\mapsto\mathbb{R}$ is a function
that assigns a real number $S(\ph,\state)$ to each probability report $\ph
\in [0,1]$ and state $\state \in \{0,1\}$. The expected value of
$S(\ph,\state)$, when $\state$ is drawn from a Bernoulli distribution with
probability $p$, is denoted by $S(\ph ; p)$. A scoring rule $S$ is
\emph{(strictly) proper} if, for every $p$, $S(\ph ; p)$ is (uniquely)
maximized at $\ph = p$.
\end{define}

Traditionally, proper binary scoring rules are used to truthfully extract the
probability of an observable binary event from an agent who knows this
probability: It is enough to pay the agent $S(\ph, \state)$ when the agent
reports the probability $\ph$ and the state turns out to be $\state$. In our
setting, obtaining truthful reports is not so straightforward: A bidder's
report affects whether or not they win the auction as well as any scoring rule
payment. However, the following theorem shows that, when the consumer welfare
function $g$ is convex, then a careful use of proper binary
scoring rules yields an incentive-compatible auction mechanism.

\begin{theorem} 
\label{thm:mech}
Let $g: \mathbb{R} \to \mathbb{R}$ be a convex function. Then there is
a truthful auction that picks the bidder $i^*$ that maximizes $v_{i^*}
+ g(p_{i^*})$ as the winner.
\end{theorem}

The proof of this theorem relies on the following lemma about proper
binary scoring rules, which is well known and given, for example, in
\cite{gneiting2007}. For the sake of completeness, we include a proof
here.
\begin{lemma} 
\label{lemma:S-proper} 
Let $g: [0,1] \to \mathbb{R}$ be a (strictly) convex
function. Then there is a (strictly) proper binary scoring rule $S_g$ such
that for every $p$, $S_g(p; p) = g(p)$. 
\end{lemma}
\begin{proof}
Let $g\prm(p)$ be a subgradient of $g$ at point $p$, i.e., a value
such that for all $q \in [0,1]$, $g(q) \geq g(p) +
g\prm(p)(q-p)$.\footnote{If $g$ is differentiable, $g\prm$ must be the 
  derivative of $g$. Even if $g$ is not differentiable, convexity of
  $g$ implies that a subgradient $g\prm$ always exists.} Now define:
\begin{align*}
  S_g(p, 1) &= g(p) + (1-p) g\prm(p)  \\
  S_g(p, 0) &= g(p) - p g\prm(p)  ~.
\end{align*}
We have that $S_g(p ; p) = pS_g(p, 1) + (1-p)S_g(p, 0) = g(p)$.
We now prove that $S_g$ is a (strictly) proper binary scoring rule. We have

\begin{eqnarray*}
S_g(\ph; p) &=& p \left(g(\ph) + (1-\ph) g\prm(\ph)\right) +
(1-p)\left(g(\ph) - \ph g\prm(\ph)\right)\\
&=&g(\ph) + (p-\ph) g\prm(\ph) . \\
\end{eqnarray*}
By definition of $g\prm$, the above value is never greater than
$g(p)$, and is equal to $g(p)$ at $\ph=p$. Therefore, the maximum of
$S_g(\ph, p)$ is achieved at $\ph = p$. If $g$ is strictly convex,
the inequality is strict whenever $\ph \neq p$.
\end{proof}

\begin{proof}[Proof of Theorem~\ref{thm:mech}]
Let $h$ be the following ``adjusted value'' function: $h(\vh,\ph) =
\vh + g(\ph)$. For convenience, rename the bidders so that bidder $1$
has the highest adjusted value, bidder $2$ the next highest, and so
on. The mechanism deterministically gives the slot to bidder $1 =
i^*$.
All bidders except bidder $1$ pay
zero. Bidder $1$ pays $h(\vh_2,\ph_2) - S_g(\ph_1,\state)$, where $S_g$
is a proper binary scoring rule satisfying $S_g(p ; p) = g(p)$ and $\state$ is
$1$ if the customer purchases the coupon and $0$ otherwise. The
existence of this binary scoring rule is guaranteed by Lemma~\ref{lemma:S-proper}.

We now show that the auction is truthful.
If $i$ bids truthfully and does not win, $i$'s utility is zero.
If $i$ bids truthfully and wins, $i$'s expected utility is
\begin{align*}
    &~v_i - h(\vh_2,\ph_2) + S_g(p_i ; p_i)  \\
  = &~h(v_i,p_i) - h(\vh_2,\ph_2) .
\end{align*}
This expected utility is always at least $0$ because $i$ is selected as
winner only if $h(v_i,p_i) \geq h(v_2,p_2)$.
This shows that the auction is interim individually rational.

Now suppose that $i$ reports $(\vh_i,\ph_i)$. If $i$ does not win the auction
with this report, then $i$'s utility is zero, but a truthful report always gives
at least zero. So we need only consider the case where $i$ wins the auction
with this report. Then, $i$'s expected utility is
\begin{align*}
      &~v_i - h(\vh_2,\ph_2) + S_g(\ph_i ; p_i)  \\
 \leq &~v_i - h(\vh_2,\ph_2) + S_g(p_i ; p_i)  \\
 =    &~h(v_i,p_i) - h(\vh_2,\ph_2) .
\end{align*} 
using the properness of $S_g$ and the definition of $h(v_i,p_i)$.
There are two cases. First, if $h(v_i,p_i) < h(\vh_2,\ph_2)$, then $U(\vh_i,\ph_i) < 0$.
But, if $i$ had reported truthfully, $i$ would have gotten a utility
of zero (having not have been selected as the winner). Second, if
$h(v_i,p_i) \geq h(\vh_2,\ph_2)$, then $U(\vh_i,\ph_i) \leq h(v_i,p_i) - h(\vh_2,\ph_2)$.
But, if $i$ had reported truthfully, $i$ would have gotten
an expected utility of $h(v_i,p_i) - h(\vh_2,\ph_2)$.
This shows incentive compatibility.
\end{proof}

%% file: impossibility.tex
\section{Impossibility Results}
\label{sec:impossibility}

An alternative way to combine consumer welfare with the
advertiser/auctioneer welfare is to ask for an outcome that maximizes
the advertiser/auctioneer welfare subject to the winner's quality
parameter meeting a minimum threshold $\alpha$. It
is not hard to show that achieving such ``discontinuous'' objective
functions is impossible.\footnote{Intuitively, the reason is that it
  is impossible to distinguish between a coin whose probability of
  heads is $\alpha$ and one whose probability is $\alpha - \epsilon$,
  when $\epsilon$ can be arbitrarily small, by the result of a single
  flip.} A more reasonable goal is
to obtain an incentive-compatible mechanism with the following
property: for two given thresholds $\alpha$ and $\beta$ with $\alpha <
\beta$, the mechanism always selects a winner $i^*$ with quality
$p_{i^*}$ at least $\alpha$, and with a value $v_{i^*}$ that is at
least $v^* := \max_{i: p_i \ge \beta}\{v_i\}$ (or an approximation of
$v^*$).

One approach to solving this problem is to use the result of the
previous section (Theorem~\ref{thm:mech}) with an appropriate choice
of the function $g$. Indeed, if we assume the values are from a
bounded range $[0,V_{\max})$ and use the auction mechanism from
Theorem~\ref{thm:mech} with a function $g$ defined as follows,

$$g(p) = \begin{cases}
0 & \mbox{if }p < \alpha\\
\frac{p - \alpha}{\beta - \alpha}.V_{\max} & \mbox{if }p\ge \alpha
\end{cases}$$
then if there is at least one bidder with quality parameter at least
$\beta$, then the mechanism is guaranteed to pick a winner with
quality at least $\alpha$. This is easy to see: the adjusted bid of
the bidder with quality at least $\beta$ is at least $V_{\max}$, while
the adjusted bid of any bidder with quality less than $\alpha$ is less
than $V_{\max}$. In terms of the value, however, this mechanism cannot
provide any multiplicative approximation guarantee, as it can select a
bidder with quality $1$ and value $0$ over a bidder with quality
$\beta$ and any value less than
$\frac{1-\alpha}{\beta-\alpha}V_{\max}$.

Unfortunately, as we show in Theorem \ref{thm:impossibility}:, this is
unavoidable: unless $\beta = 1$ (that is, unless welfare is compared
only against bidders of ``perfect'' quality), there is no deterministic,
truthful mechanism that can guarantee a bounded multiplicative
approximation guarantee in the above setting.

\begin{theorem} 
\label{thm:impossibility}
For a given $0 \leq \alpha < \beta \leq 1$ and $\lambda \geq 1$,
suppose that a deterministic truthful mechanism satisfies that, if there
is some bidder $i$ with $p_i \geq \beta$:
\begin{enumerate}[topsep=2pt]
 \item The winner has $p_{i^*} > \alpha$;
 \item The winner has value $v_{i^*} \geq v^* / \lambda$, where
       $v^* := \max_{i: p_i \ge \beta}\{v_i\}.$
\end{enumerate}
Then $\beta = 1$. This holds even if valuations are upper-bounded by a
constant $V_{max}$.
\end{theorem}
\begin{proof}
Fix all reports $\vv_{-i}$ and $\pv_{-i}$. Let
$t_{\state}(\vh_i,\ph_i)$ be the net transfer to bidder $i$ in state
$\state$ when $i$ reports $(\vh_i,\ph_i)$ and wins the auction
($t_{\state}(\vh_i,\ph_i)$ will be negative if the mechanism charges
bidder $i$). Then we can denote $i$'s expected utility for winning
with report $(\vh_i,\ph_i)$ given true type $(v_i,p_i)$ by
  \[ U(\vh_i,\ph_i ; v_i,p_i) = v_i + p_i t_1(\vh_i,\ph_i) + (1-p_i) t_0(\vh_i,\ph_i)  ~.\]
Since $i$ will always report so as to maximize this value given that $i$
prefers to win, we can define
  \[  h(p_i) = \max_{(\vh_{i},\ph_{i})\in \mathcal{W}} \{p_it_1(\vh_i,\ph_i) + (1-p_i)t_0(\vh_i,\ph_i)\},  \] where $\mathcal{W}$ is the set of winning bids $(\vh_{i},\ph_{i})$,
and write $i$'s expected utility for winning simply as
$U(v_i,p_i) = \max_{(\vh_i,\ph_i)\in\mathcal{W}}U(\vh_i,\ph_i ; v_i,p_i)
= v_i + h(p_i)$.
We note that $h$ is a convex function of $p$ since, for any pricing scheme
$t_{\state}(\vh_i,\ph_i)$, $h(p)$ is the point-wise maximum over a family
of linear functions.

Fix some choices of $0 \leq \alpha < \beta \leq 1$. To guarantee that $i$
does not win if $p_i \leq \alpha$, we must have that, whenever $p_{i} \leq \alpha$,
every winning bid gives $i$ negative expected utility. Therefore, $i$ will not bid
so as to win in this case. Thus,
\begin{align*}
  U(v_i,p_i)         &< 0  &(\forall v_i, p_{i} \leq \alpha)   \\
  \implies h(\alpha) &< -V_{max}  ~.
\end{align*}

Now, suppose there is a $v_1$ with the property that $i$ is never selected
as winner when $v_i < v_1$. Then we must have
\begin{align*}
  U(v_i,p_i)         &< 0  &(\forall v_i < v_1, p_i)  \\
  \implies h(p_i)      &< -v_1  &(\forall p_i) ~.
\end{align*}

Conversely, suppose that there is a $v_2$ with the property that $i$ is
always selected as the winner when $p_i \geq \beta$ and $v_i > v_2$. Then we must have
\begin{align*}
  U(v_i,p_i)         &> 0  &(\forall v_i > v_2, p_i \geq \beta) \\
  \implies h(\beta)  &> -v_2   ~.
\end{align*}

Since $h$ is convex,
\begin{align*}
 \left(\frac{\beta - \alpha}{1-\alpha}\right)h(1) + \left(\frac{1-\beta}{1-\alpha}\right)h(\alpha) &\geq h\left(\frac{\beta - \alpha}{1-\alpha} + \frac{1-\beta}{1-\alpha}\left(\alpha\right)\right)  \\
   &= h(\beta) ~.
\end{align*}
The above inequalities thus imply that
\begin{align}
 v_2 \geq V_{max} \left(\frac{1-\beta}{1-\alpha}\right) + v_1 \left(\frac{\beta - \alpha}{1-\alpha}\right)  ~.  \label{eqn:v1v2-ineq}
\end{align}
Now suppose that our mechanism guarantees a welfare approximation factor
of $\lambda$. Let $v^*$ be the highest value of any bidder other than $i$
having $p \geq \beta$ (supposing such a bidder exists). Then $i$ loses if
$v_i < v^* / \lambda = v_1$ and wins whenever $p_i \geq \beta$ and $v_i > \lambda v^* = v_2$.
But $v_1$ and $v_2$ satisfy the properties given above, so they satisfy
Inequality \ref{eqn:v1v2-ineq}. Now take $v^*$ arbitrarily small, so
that $v_1,v_2 \ll V_{max}$,
and Inequality \ref{eqn:v1v2-ineq} can only hold if $\beta = 1$.
\end{proof}

The techniques used in the above proof can be used to show that the
convexity assumption in Theorem~\ref{thm:mech} is indeed necessary:

\begin{theorem}
\label{thm:convexitynecessary}
Assume $g: \mathbb{R} \to \mathbb{R}$ is a function for which there
exists a deterministic truthful auction that always picks the bidder
$i^*$ that maximizes $v_{i^*} + g(p_{i^*})$ as the winner. Then $g$ is
a convex function.
\end{theorem}
\begin{proof}
As in the proof of Theorem~\ref{thm:impossibility}, fix all reports
$\vv_{-i}$ and $\pv_{-i}$, and define $t_{\state}(\vh_i,\ph_i)$ and
$U(\vh_i,\ph_i ; v_i,p_i)$ as before. By the incentive compatibility
and individual rationality of the mechanism, bidder $i$ must win the
auction if

$$\max_{(\vh_i,\ph_i)} U(\vh_i,\ph_i ; v_i,p_i) > 0$$
and lose if $$\max_{(\vh_i,\ph_i)} U(\vh_i,\ph_i ; v_i,p_i) < 0.$$
Equivalently, bidder $i$ must win if 
$$v_i > - \max_{(\vh_i,\ph_i)} \{p_it_1(\vh_i,\ph_i) + (1-p_i)t_0(\vh_i,\ph_i)\}$$ 
and lose if the opposite inequality holds. On the other hand, since
the mechanism always picks the bidder that maximizes $v_i + g(p_i)$,
bidder $i$ must win if 
$$v_i > \max_{j\neq i}\{v_j + g(p_j)\} - g(p_i)$$
and lose if the opposite inequality holds. Thus, we must have:

\begin{align*}
&\max_{j\neq i}\{v_j + g(p_j)\} - g(p_i) =\\
&\qquad - \max_{(\vh_i,\ph_i)} \{p_it_1(\vh_i,\ph_i) + (1-p_i)t_0(\vh_i,\ph_i)\},
\end{align*}
or
$$g(p_i) = \max_{(\vh_i,\ph_i)} \{p_it_1(\vh_i,\ph_i) +
(1-p_i)t_0(\vh_i,\ph_i)\} + \max_{j\neq i}\{v_j + g(p_j)\}.$$ 

The right-hand side of the above equation is the maximum of a number
of terms, each of which is a linear function of $p_i$. Therefore,
$g(p_i)$ is a convex function of $p_i$.
\end{proof}

%% file: general.tex
\section{A General Framework} 
\label{sec:general}

Daily deals websites generally offer many deals simultaneously, and to
many consumers. A more realistic model of this scenario must take into
account complex \emph{valuation functions} as well as general \emph{quality reports}.
Merchants' valuations may depend on which slot (top versus bottom,
large versus small) or even \emph{subset} of slots they win; they may also
change depending on which competitors are placed in the other slots. Meanwhile,
merchants might like to report quality in different units than purchase probability,
such as (for example) total number of coupon sales in a day, coupon sales
relative to those of competitors, or so on.

In this section, we develop a general model that can cover these cases
and considerably more. As in a standard multidimensional
auction, bidders have a valuation for each outcome of the mechanism (for
instance, each assignment of slots to bidders). For quality reports, our key
insight is that they may be modeled by a \emph{belief} or \emph{prediction}
over possible states of the world, where each state has some verifiable quality.
This naturally models many scenarios where the designer
would like to make a social choice (such as allocating goods) based
not only on the valuations of the agents involved, but also on
the likely externality on some non-bidding party; however, this externality
can be best estimated by the bidders. We model this externality by
a function, which we call the \emph{consumer welfare function},
that maps probability distributions to a welfare value.
A natural consumer welfare function is the expected value of
a distribution.

When this consumer welfare satisfies a convexity condition, we construct
truthful mechanisms for welfare maximization in this general setting; we also prove matching negative
results. This allows us to characterize implementable welfare functions in terms
of \emph{component-wise convexity}, which includes the special case of
expected value and can also capture intuitively risk-averse preferences.

We start with a definition of the model in
Section~\ref{subsection:general-model}, and then give a truthful
mechanism as well as a matching necessary condition for
implementability in this model in
Section~\ref{subsection:truthful-general}. In
Section~\ref{subsection:examples-general} we give a number of applications
and extensions of our general framework.

\subsection{Model} 
\label{subsection:general-model}

We now define the general model, using the multi-slot daily deals
problem as a running example to
illustrate the definition.

There are $m$ {\em bidders} (also called {\em merchants}) indexed $1$
through $m$, and a finite set $\Outcomes$
of possible {\em outcomes} of the mechanism. Each bidder has as private
information a valuation function $v_i : \Outcomes \to \mathbb{R}$ that
assigns a value $v_i(o)$ to each outcome $o$. For instance, each outcome
$\outcome$ could correspond to an assignment of
merchants to the available slots, and $v_i(o)$ is $i$'s expected value
for this assignment, taking into account the slot(s) assigned to $i$ as
well as the coupons in the other slots.

For each $\outcome \in \Outcomes$ and each
bidder $i$, there is a finite set of observable disjoint {\em states}
of interest $\ioStates$ representing different events that could occur
when the mechanism's choice is $o$. For example, if merchant $i$ is awarded a slot under
outcome $\outcome$, then $\ioStates$ could be the possible total numbers of
sales of $i$'s coupon when the assignment is $o$, \emph{e.g.} $\ioStates = \{$fewer than $1000$,
$1000$ to $5000$, more than $5000\}$.

Given an outcome $o$ chosen by the mechanism, nature will select at random
one of the states $\omega$ in $\ioStates$ for each bidder $i$.\footnote{These choices do
not have to be independent across bidders; indeed, all bidders could be predicting
the same event, in which case $\ioStates = \ioprmStates$ for all $i,i\prm$ and
nature selects the same state for each $i$.} In the running example,
some number of consumers choose to purchase $i$'s coupon, so perhaps
$\omega = \text{``$1000$ to $5000$''}$.

We let $\Delta_{\ioStates}$ denote the probability simplex over the
set $\ioStates$, i.e., $\Delta_{\ioStates} = \{p\in [0,1]^{\ioStates}:
\sum_{\omega\in\ioStates} p_{\omega} = 1\}$. Each bidder $i$ holds as private
information a set of beliefs (or predictions) $p_i:
\Outcomes \to \Delta_{\ioStates}$. For each outcome $\outcome$,
$p_i(\outcome) \in \Delta_{\ioStates}$ is a probability distribution
over states $\state \in \ioStates$. Thus, under outcome $o$ where $i$ is
assigned a slot, $p_i(o)$ would give the probability that
$i$ sells fewer than $1000$ coupons, that $i$ sells between $1000$ and $5000$ coupons,
and that $i$ sells more than $5000$ coupons. We denote the vector
of predictions $(p_1(\outcome),\hdots,p_m(\outcome))$ at outcome $o$
by $\pv(\outcome)\in \cart_{i=1}^m\Delta_{\ioStates}$.

The goal of the mechanism designer is to pick an outcome that
maximizes a notion of welfare. The combined welfare of the bidders and
the auctioneer can be represented by $\sum_{i=1}^m v_i(o)$. If this
was the goal, then the problem could have been solved by ignoring
the $p_i(o)$'s and using the
well-known Vickrey-Clarke-Groves
mechanism~\cite{Vickrey,Clarke,Groves}. In our setting, however, there is another component in the
welfare function, which for continuity with the daily deals setting
we call the {\em consumer welfare}. This component, which depends
on the probabilities $p_i(o)$, represents the
welfare of a non-bidding party that the auctioneer wants to keep
happy (which could even be the auctioneer herself!). The consumer welfare
when the mechanism chooses outcome $o$ is given by an arbitrary function
$g_o: \cart_{i=1}^m \Delta_{\ioStates} \to \mathbb{R}$ which depends on
the bidders' predictions $\pv(\outcome)$. The goal of the
mechanism designer is then to pick an outcome $o$ that maximizes 
$$\left( \sum_{i=1}^m v_i(o) \right) + g_{\outcome}(\pv(\outcome)).$$

For example, in the multi-slot problem, consumer welfare at the
outcome $o$ could be defined as the sum of the expected number of clicks of
the deals that are allocated a slot in $o$. 

A mechanism in this model elicits bids $(\vh_i, \ph_i)$ from each
bidder $i$ and picks an outcome $o$ based on these bids. Then, for each
$i$, the mechanism observes the state $\omega_i$ picked by nature from
$\ioStates$ and charges $i$ an amount that can depend on
the bids as well as the realized state $\omega_i$. This mechanism is
{\em truthful} (incentive compatible and
individually rational) if, for each bidder $i$, and for any set of
reports of other bidders $(\vh_{-i}, \ph_{-i})$, bidder $i$ can
maximize her utility by bidding her true type $(v_i, p_i)$, and this
utility is non-negative.

\subsection{A Class of Truthful Mechanisms} 
\label{subsection:truthful-general}
In this section, we give a truthful mechanism for the general setting,
assuming that the consumer welfare function $g_o$ satisfies the following
convexity property.
\begin{define}
A function $f: \Delta_{\States}\mapsto\mathbb{R}$ is convex if and
only if for each $x, y\in\Delta_{\States}$ and each $\alpha\in[0,1]$, 
$$f(\alpha x + (1-\alpha)y)\le \alpha f(x) + (1-\alpha)f(y).$$ We call
a function $g_{\outcome}: \cart_{i=1}^m\Delta_{\ioStates}\mapsto
\mathbb{R}$ {\em component-wise convex} if for each $i$ and for each
vector $\pv_{-i}(o)\in \cart_{j: j\neq i}\Delta_{\States_{j,o}}$ of
predictions of bidders other than $i$, $g_o(p_i(o), \pv_{-i}(o))$ is a
convex function of $p_i(o)$.
\end{define}

Component-wise convexity includes the important special case of
expected value, and can also capture an intuitive notion of risk aversion with respect
to each bidder's prediction, as it requires that the value of taking
a draw from some distribution gives lower utility than the expected
value of that draw (see footnotes \ref{ft:risk-example} and \ref{ft:risk-concave}).
It includes all convex functions, but there are also functions such as
$g(p_1,p_2) = \left(p_1 - \frac{1}{2}\right) \cdot \left(p_2 - \frac{1}{2}\right)$
that are component-wise convex but not convex.

We can now state our result.

\begin{theorem}
\label{theorem:truthful-general}
If for any outcome $o$, the consumer welfare function $g_o$ is
component-wise convex, then there is a truthful mechanism that selects
an outcome $o$ that maximizes $\left(\sum_{i=1}^m v_i(o)\right) +
g_{\outcome}(\pv(\outcome)).$
\end{theorem}

As in the simple model, our mechanism uses proper scoring rules.  The
definition of scoring rules can be adapted to the general setting as
follows.

\begin{define}
A \emph{scoring rule} $S : \Delta_{\States} \times \States \to
\mathbb{R}$ is a function that assigns a real number $S(p,\state)$ to
each probability report $p\in \Delta_\States$ and state
$\state\in\States$. The expected value of $S(\ph,\omega)$ when
$\omega$ is drawn according to the distribution $p\in\Delta_\States$
is denoted by $S(\ph ; p)$.  A scoring rule $S$ is \emph{proper} if,
for every $p$, $S(\ph ; p)$ is maximized at $\ph = p$.
\end{define}

We also need a generalization of Lemma \ref{lemma:S-proper}. This fact
about scoring rules is proven in more generality
in~\cite{gneiting2007} and is originally due to
Savage~\cite{savage1971}.

\begin{lemma}[\cite{gneiting2007,savage1971}]
\label{lemma:S-proper-general}
For every convex function $g : \Delta_{\States} \to \mathbb{R}$ there
is a proper scoring rule $S_g$ such that for every $p$, $S_g(p;p) =
g(p)$.
\end{lemma}

\begin{proof}[Proof of Theorem \ref{theorem:truthful-general}]
Using Lemma \ref{lemma:S-proper-general} and the assumption that $g_o$
is component-wise convex, for any outcome $o$, bidder $i$, and set of
reports of other bidders $\pv_{-i}(o)$, we can construct a proper
scoring rule $\Sio$ with $\Sio(p_i(\outcome) ; p_i(\outcome)) =
g_o(p_i(o), \pv_{-i}(o))$. The function $\Sio$ scores prediction
$\ph_i(\outcome)$ on states $\state \in \ioStates$.

Next, we use a Vickrey-Clark-Groves-like mechanism and show that
truthfulness is a dominant strategy.  Let
$W^o = \sum_{i=1}^m v_i(o) + g_{\outcome}(\pv(\outcome)),$
where $(v_i,p_i)$ is the bid of bidder $i$.  Our mechanism selects
outcome $o^*$ that maximizes $W^{o^*}$. Let $W_{-i}$ be the value of
the selection made by our mechanism on the set of bids excluding $i$.
Each bidder $i$, when state $\state \in \States_{i,o^*}$ occurs, pays
$$W_{-i} - \mysum{i\prm \neq i}{} v_{i\prm}(o^*) - S_{o^*, i, \pv_{-i}(o^*)}(p_i(o^*),\state).$$
Therefore, under outcome $o^*$, bidder $i$'s expected utility for
reporting truthfully is
\begin{align}
\label{eqn:1}
 U(v_i,p_i) =   &~ v_i(o^*) - W_{-i} + \mysum{i\prm \neq i}{} v_{i\prm}(o^*)
+ S_{o^*, i, \pv_{-i}(o^*)}(p_i(o^*); p_i(o^*)) \nonumber\\
= & \mysum{i\prm=1}{m} v_{i\prm}(o^*)- W_{-i} + g_{o^*}(\pv(o^*)) \nonumber\\
=   &~ W^{o^*} - W_{-i}
\end{align}

The above value is clearly non-negative. Now, consider a scenario
where $i$ changes her bid to $(\vh_i,\ph_i)$, when her type is still
given by $(v_i,p_i)$. Let $o'$ denote the outcome selected in that
scenario. The utility of bidder $i$ in this scenario can be written as:
\begin{align}
\label{eqn:2}
 U(\vh_i,\ph_i) =   &~ v_i(o') - W_{-i} + \mysum{i\prm \neq i}{} v_{i\prm}(o')
+ S_{o', i, \pv_{-i}(o')}(\ph_i(o'); p_i(o')) \nonumber\\
\le & \mysum{i\prm=1}{m} v_{i\prm}(o')- W_{-i} + g_{o'}(\pv(o')) \nonumber\\
=   &~ W^{o'} - W_{-i},
\end{align}
where the first inequality follows from the fact that $S_{o', i,
  \pv_{-i}(o')}$ is a proper scoring rule with $S_{o', i,
  \pv_{-i}(o')}(p_i(o'); p_i(o')) = g_{o'}(\pv(o'))$.

Finally, note that by the definition of $o^*$, we have $W^{o^*}\ge
W^{o'}$. This inequality, together with \eqref{eqn:1} and
\eqref{eqn:2} implies that $U(\vh_i,\ph_i)\le U(v_i,p_i)$. Therefore,
$i$ cannot gain by misreporting her type.
\end{proof}

\subsection{Characterization of Implementable Consumer Welfare Functions}
\label{subsection:impossibility-general}

In this section, we show that the {\em component-wise convexity}
assumption that we imposed on the consumer welfare function in the
last section to derive a truthful mechanism is indeed necessary. In
other words, in our general setting, the consumer welfare functions
that are implementable using dominant-strategy truthful mechanisms are
precisely those that are component-wise convex.

\begin{theorem}
\label{thm:necessary}
Suppose $g$ is a consumer welfare function and there exists a
deterministic truthful mechanism that always selects the outcome $o$
that maximizes $\sum_{i=1}^{m} v_i(o) + g_o(\pv(o))$.  Then $g_o$ is
component-wise convex for every $o$.
\end{theorem}

\begin{proof}
Fix a bidder $i$ and bids $(v_{-i}, p_{-i})$ of all the other
bidders. We prove that for every outcome $o$, the function
$g_o(\pv(o))$ as a function of $p_i(o)$ is convex. This shows that
$g_o$ is component-wise convex for every $o$.

For any bid $(\vh_i,\ph_i)$ for bidder $i$, the mechanism selects an
outcome $o$, and charges $i$ an amount depending on the realized state
$\state\in\ioStates$. This payment can be represented by a vector in
$\mathbb{R}^{\ioStates}$ (a negative value in this vector indicates a
value that the bidder pays the auctioneer, and a positive value
indicates a reverse transfer). Let $A_{o,i,v_{-i},p_{-i}}\subseteq
\mathbb{R}^{\ioStates}$ denote the collection of payment vectors
corresponding to all bids $(\vh_i,\ph_i)$ for bidder $i$ that (along
with the bids $(v_{-i}, p_{-i})$ for others) result in the mechanism
picking outcome $o$. Since we have fixed $i$ and $v_{-i},p_{-i}$, we
simply denote this collection by $A_o$.

The utility of $i$ when she submits a bid that results in outcome $o$
and payment $t\in A_o$ can be written as $v_i(o) + t.p_i(o)$ (the
latter term is the inner product of $t\in \mathbb{R}^{\ioStates}$ and
$p_i(o)\in\Delta_{\ioStates}$). By truthfulness of the mechanism,
$i$'s utility from truthful bidding must be

$$\max_{o\in\Outcomes}\max_{t\in A_o}\{v_i(o) + t.p_i(o)\},$$
and the outcome selected by the mechanism must be the $o$ that maximizes
the above expression. Denoting
\begin{equation}
\label{eqn:fo}
f_o(p_i(o)) = \max_{t\in A_o}\{t.p_i(o)\},
\end{equation}
this means that the mechanism selects the outcome $o$ if
$$v_i(o) + f_o(p_i(o)) > \max_{o'\neq o}\{v_i(o') +
f_{o'}(p_i(o'))\}$$ and does not select this outcome if the reverse
inequality holds. This means that holding everything other than
$v_i(o)$ constant, the threshold for $v_i(o)$ after which the
mechanism selects the outcome $o$ is precisely
\begin{equation}
\label{eqn:thresh1}
\max_{o'\neq o}\{v_i(o') + f_{o'}(p_i(o'))\} - f_o(p_i(o)).
\end{equation}

On the other hand, the mechanism always picks an outcome $o$ that
maximizes $\sum_{j=1}^{m} v_j(o) + g_o(\pv(o))$. Therefore, holding
everything except $v_i(o)$ constant, the threshold for $v_i(o)$ after
which the outcome $o$ is selected is precisely
\begin{equation}
\label{eqn:thresh2}
\max_{o'\neq o}\left\{\mysum{j=1}{m} v_j(o') + g_{o'}(\pv(o'))\right\} -
\mysum{j\neq i}{} v_j(o) - g_o(\pv(o)).
\end{equation}

Therefore, the thresholds~\eqref{eqn:thresh1} and~\eqref{eqn:thresh2}
must be equal. Writing this equality, and moving $g_o(\pv(o))$ to the
left-hand side of the equality and everything else to the right-hand
side, we obtain:

\begin{align*}
g_o(\pv(o)) = & - \max_{o'\neq o}\{v_i(o') + f_{o'}(p_i(o'))\} + f_o(p_i(o)) \\&+\,
\max_{o'\neq o}\left\{\mysum{j=1}{m} v_j(o') + g_{o'}(\pv(o'))\right\}
- \mysum{j\neq i}{} v_j(o).
\end{align*}

Now, observe that the only term on the right-hand side of the above
equation that depends on $p_i(o)$ is $f_o(p_i(o))$. Furthermore,
$f_o(p_i(o))$ (as defined in Equation~\eqref{eqn:fo}) is the maximum
over a family of linear functions of $p_i(o)$, and therefore is a
convex function of $p_i(o)$. This means that fixing any set of values
for $\pv_{-i}(o)$, $g_o(\pv(o))$ is a convex function of
$p_i(o)$. Therefore $g_o$ is component-wise convex.
\end{proof}

\subsection{Applications}
\label{subsection:examples-general}

In this section, we present a few sample applications and extensions
of our general framework. This demonstrates that the results of
Section~\ref{subsection:truthful-general} can be used to
characterize achievable objective functions and design truthful
mechanisms in a very diverse range of settings.

\medskip
\noindent
{\bf Daily Deals with Both Merchant and Platform Information.} In some cases,
it might be reasonable in a daily deals setting to suppose that the platform, as well as the merchant,
has some relevant private information about deal quality. For example, perhaps
the merchant has specific information about his particular deal, while the
auctioneer has specific information about typical consumers under particular
circumstances (days of the week, localities, and so on).
Many such extensions are quite straightforward; intuitively, this is because
we solve the difficult problem: incentivizing merchants to truthfully reveal quality information.

To illustrate, consider a simple model where merchant $i$ gets utility $a_i$ from
displaying a deal to a consumer and an additional $c_i$ if the user purchases it.
For every assignment of slots $o$ containing the merchant's deal, its quality (probability of purchase)
is a function $f_{o,i}$ of two pieces of private information: $x_i$, held by the merchant,
and $y_i$, held by the platform. Each
merchant is asked to submit $(a_i, c_i, x_i)$. The platform computes, for each slot assignment $o$,
$p_i(o) = f_{o,i}(x_i,y_i)$, then sets $v_i(o) = a_i + p_i(o) c_i$ for all $o$
that include $i$'s coupon ($v_i(o) = 0$ otherwise). Then,
the platform runs the auction defined in Theorem \ref{theorem:truthful-general}, setting $i$'s bid equal to $(v_i,p_i)$.
By Theorem \ref{theorem:truthful-general}, bidder $i$ maximizes expected utility when $v_i$ is her
true valuation for winning and $p_i$ is her true deal quality; therefore,
she can maximize expected utility by truthfully submitting $(a_i, c_i, x_i)$,
as this allows the mechanism to correctly compute $v_i$ and $p_i$.

\medskip
\noindent
{\bf Reliable Network Design. } Consider a graph $G$, where each edge
is owned by a different agent.  The auctioneer wants to buy a path from
a source node $s$ to a destination node $t$.
Each edge has a cost for being used in the path, and also a
probability of failure. Both of these parameters are
private values of the edge. The goal of the mechanism designer is to
buy a path from $s$ to $t$ that minimizes the total cost of the edges
plus the cost of failure, which is a fixed constant times the probability
that at least one of the edges on the path fails.

It is easy to see that the above problem fits in our general
framework: each bidder's value is the negative of the cost of the
edge; each ``outcome'' is a path from $s$ to $t$; for each edge $i$ on
a path, the corresponding ``states'' are fail and succeed; the
consumer welfare function $g_o$ for an outcome $o$ is the negative of
the failure cost of that path. For each edge, fixing all other reports,
$g_o$ is a linear function of failure probability. Therefore,
$g_o$ is component-wise convex, and
Theorem~\ref{theorem:truthful-general} gives a
truthful mechanism for this problem.

We can also model a scenario where each edge has a probabilistic delay
instead of a failure probability. When edge $i$ is included in the path,
the possible states $\ioStates$ correspond to the possible delays experienced
on that edge. A natural objective function is to minimize the total
cost of the path from $s$ to $t$ plus its expected delay, which is a
linear function of probability distributions.
We can also implement costs that are concave functions of the
delay on each edge (as welfare, the negative of cost, is then convex).
These model risk aversion, as, intuitively, the cost of a delay drawn
from a distribution is higher than the cost of the expected delay of
that distribution. (Note that our results imply that a \emph{concave}
objective function is \emph{not} implementable!)

The exact same argument shows that other network design problems fit in our
framework. For example, the goal can be to pick a $k$-flow from $s$ to
$t$, or a spanning tree in the graph. The ``failure'' function can
also be more complicated, although we need to make sure the convexity
condition is satisfied. 

\medskip
\noindent
{\bf Principal-Agent Models with Probabilistic Signals. } Another
application of our mechanism is in a \emph{principal-agent} setting,
where a principal would like to incentivize agents to exert an optimal
level of effort, but can only observe a probabilistic signal of
this effort. Suppose the principal wishes to hire a set of agents
to complete a project; the principal only observes whether each agent
succeeds or fails at his task, but the probability of each's success is
influenced by the amount of effort he puts in. More precisely, let
$c_i(e)$ denote the cost of exerting effort $e$
for agent $i$ and $p_i(e)$ denote the probability of the agent's success
if this agent is hired and exerts effort $e$. The welfare generated by
the project is modeled by a component-wise convex function of the
agents' probabilities of success (for instance,
a constant times the probability that all agents succeed).

At the first glance, it might seem that this problem does not fit
within our framework, since each agent can affect its success
probability by exerting more or less effort. However, suppose we
define an outcome of the mechanism as selecting both a set of agents
and an assignment of effort levels to these agents.
Each agent submits as his type the cost $c_i(e)$ and probability of
success $p_i(e)$ for each possible effort level $e$.
Theorem~\ref{theorem:truthful-general} then gives a welfare-maximizing
mechanism that truthfully elicits the $c_i(e)$ and $p_i(e)$ values from each agent and
selects the agents to hire, the effort levels they should exert, and
the payment each receives conditional on whether his component of
the project succeeds.
Agents maximize expected utility by declaring
their true types and exerting the amount of effort
they are asked to.\footnote{The proof is the same as that of truthfulness:
If the agent deviates and exerts some other effort level, his
expected utility will be bounded by if he had reported the truth and the
mechanism had assigned him that effort level; but by design, this is
less than his utility under the choice actually made by the mechanism.}

%% file: conclusion.tex
\section{Conclusion}

Markets for daily deals present a challenging new mechanism-design
setting, in which a mechanism designer (the platform) wishes
to pick an outcome (merchant and coupon to display) that not only
gives good bidder/auctioneer welfare, but also good welfare for
a third party (the consumer); however, this likely consumer welfare
is private information of the bidders.

Despite the asymmetry of information, we show that, when
the consumer welfare function is a convex function of bidders'
quality, we can design truthful
mechanisms for social welfare maximization in this setting.
We give a matching negative result showing that no truthful,
deterministic mechanism exists when consumer welfare is not convex.
Another natural objective, approximating welfare subject to meeting
a quality threshold, also cannot be achieved in this setting.

Extending the daily deals setting to a more general domain yields a
rich setting with many potential applications. We model this setting
as an extension to traditional mechanism design: Now, agents have both
preferences over outcomes and probabilistic beliefs conditional on
those outcomes. The goal is to maximize social welfare including
the welfare of a non-bidding party, modeled by a consumer welfare
function taking probability distributions over states of the world to welfare.

A truthful
mechanism must incentivize bidders to reveal their true preferences
\emph{and} beliefs, even when these revealed beliefs influence the designer
to pick a less favorable outcome for the bidders.
We demonstrate that this is possible if and only if the consumer welfare function
is component-wise convex, and when it is, we explicitly design mechanisms
to achieve the welfare objective. Component-wise convexity includes
expected-welfare maximization and intuitively can capture risk averse preferences.
Finally, we demonstrate the generality of our results with a number of example
extensions and applications.